\documentclass[lettersize,journal]{IEEEtran}
\usepackage{amsmath,amsfonts}
\usepackage{algorithmic}
\usepackage{algorithm}
\usepackage{array}
\usepackage[caption=false,font=normalsize,labelfont=sf,textfont=sf]{subfig}
\usepackage{textcomp}
\usepackage{stfloats}
\usepackage{url}
\usepackage{verbatim}
\usepackage{graphicx}
\usepackage{cite}
\newtheorem{theorem}{Theorem}

\newtheorem{pb}{Problem}

\newtheorem{proof}{Proof}
\usepackage{color}

\begin{document}

\title{Age of Synchronization Minimization in Wireless Networks with Random Updates and Time-Varying Timeliness Requirement}

\author{Yuqiao He,Yuchao Chen, Jintao Wang,~\IEEEmembership{Senior Member,~IEEE}, Jian Song,~\IEEEmembership{Fellow,~IEEE}
        % <-this % stops a space
\thanks{Yuqiao He and Yuchao Chen is with Beijing National Research Center for Information Science and Technology (BNRist) and Department of Electronic Engineering, Tsinghua University, Beijing, China (e-mail: heyq18@mails.tsinghua.edu.cn; cyc20@mails.tsinghua.edu.cn)}% <-this % stops a space
\thanks{Jintao Wang is with Beijing National Research Center for Information Science and Technology (BNRist) and Department of Electronic Engineering, Tsinghua University, Beijing, China and Research Institute of Tsinghua University in Shenzhen, Shenzhen, China and State Key laboratory of Space Network and Communications, Tsinghua University (e-mail: wangjintao@tsinghua.edu.cn)}
\thanks{Jian Song is with Beijing National Research Center for Information Science and Technology (BNRist) and Department of Electronic Engineering, Tsinghua University, Beijing, China and Shenzhen International Graduate School, Tsinghua University, Shenzhen, Guangdong, 518055 (e-mail: jsong@tsinghua.edu.cn)}
}



\maketitle

\begin{abstract}
  This study considers a wireless network where multiple nodes transmit status updates to a base station (BS) via a shared, error-free channel with limited bandwidth. The status updates arrive at each node randomly. We use the Age of Synchronization (AoS) as a metric to measure the information freshness of the updates. The AoS of each node has a timely-varying importance which follows a Markov chain. Our objective is to minimize the weighted sum AoS of the system. The optimization problem is relaxed and formulated as a constrained Markov decision process (CMDP). Solving the relaxed CMDP by a linear programming algorithm yields a stationary policy, which helps us propose a near-stationary policy for the original problem. Numerical simulations show that when the weight state transitions are non-independent, the AoS performance of our policy outperforms the Max-Weight policy which has great AoS performance in time-invariant weight scenarios.
\end{abstract}


\section{Introduction}
\IEEEPARstart{A}{s} wireless communication technology progresses, the concept of the Internet of Things (IoT) has emerged as a significant trend. In these scenarios, such as Industrial IoT (IIoT) \cite{Sisinni2018} and Vehicle-to-everything (V2X) \cite{etsi}, a central controller must monitor numerous different type of sensors, collecting status updates from the sensors to make decisions or operations. The information that the central controller focuses on is highly time-sensitive, indicating that the value of information is directly related to its freshness. 

To evaluate the freshness of information, the most commonly used metric is the Age of Information (AoI), which is defined as the time elapsed since the most recent information was generated \cite{Kaul2012}. The AoI is mainly applied in generate-at-will scenarios where status updates can be sampled at any time. In \cite{Krikidis2019,Ceran2019,Moltafet2020,Miridakis2022}, the minimization of the system AoI in wireless networks has been studied, while the studies in \cite{Farazi2018,Wu2024} optimized the average AoI of each node in multi-user networks. In an IoT system, the timeliness requirement of different sensors or nodes may differ. Taking IIoT as an example, the sensor for fire warning needs to keep higher timeliness than the sensor monitoring inventory. So the AoI of different sensor should be assigned different weights. The minimization of the weighted sum AoI in multi-user networks has been studied in \cite{Yates2017,Kadota2019,Sun2021}.

However, the aforementioned researches have two issues: The AoI is not suitable for scenarios with random arrivals and they all focuse on fixed, time-invariant weights.

In a random-updated scenario, new packages arrive at the sensors randomly. When the system is synchronized, meaning the newest status update has been transmitted to the central controller, the AoI will still grow, which can not depict the freshness of the information. The metric that can accurately reflect this state is the Age of Synchronization (AoS), which is defined as the interval between the current time and the last time the system was synchronized \cite{Zhong2018}. In \cite{Joo2017,Tang2020,Chen2022,Gong2024}, scheduling policies minimizing the average AoS of multiple users were proposed over the time-invariant channel, and in \cite{Zhang2020}, the minimization was studied in the time-varying channel. 
Considering the different importance of the sensors, the weighted sum AoS of the network was minimized under throughput constraints in \cite{He2022}.

In practical scenarios, the weights of the sensors are typically dynamic and variable. Within an IIoT system, the interest on the information of each node may change across different production stages. For instance, production safety may be prioritized during production while fire prevention of the warehouse may be focused more during downtime. Similarly, vehicles must focus on different aspects depending on the environment in autonomous driving scenarios. Autonomous vehicles need to closely monitor road signals and surrounding pedestrians when traveling on city roads. However, the concentration is more on speed control and crisis prevention when traveling on highways. Therefore, the policies applying fixed weights may not be the most effective way to optimize the information timeliness of the system. The time-varying weights or popularities were considered in \cite{Bharath2018,Gao2021}, but their targets were not the optimization of information timeliness. Scheduling policies were designed to improve information timeliness in \cite{Kam2017} and \cite{Tang2021}, while their optimization metrics were the request missing rate and the AoI respectively. There is a lack of research on information timeliness optimization in random-updated scenarios with time-varying weights.

To address this issue, we focus on a multi-user wireless network with randomly arrived status updates and time-varying weights. Our objective is to minimize the weighted sum AoS of the system. The minimization problem is relaxed and modeled as a constrained Markov decision process (CMDP), which is solved by a Lagrange function and a linear programming algorithm. A scheduling policy is proposed and its AoS performance is shown by numerical simulations compared with the policy regardless of weight variations.

The rest of the paper is structured as follows. Section~\ref{model} presents the system model and formulates the optimization problem. In Section~\ref{resolution} the optimization problem is solved in three steps and the specific scheduling policy is given. Then in Section~\ref{simulation} numerical simulations show the performance of our policy. Finally, Section~\ref{conclusion} draws the conclusion.

\section{System Model} \label{model}
We consider a wireless network where $M$ nodes transmit randomly arrived status updates to a base station (BS). The time is slotted and denoted by $t\in\{1,2,\dots,T\}$. A new status update arrives at node $i\in\{1,2,\dots,M\}$ with probability $\lambda_i$ in each time slot. The arrival state at node $i$ is represented by an indicator function $\Lambda_i(t)$ which is equal to $1$ when a new update arrives at node $i$ in the time slot before time slot $t$ and is equal to $0$ otherwise.

The transmission status of node $i$ in time slot $t$ is denoted by $u_i(t)$, meaning that $u_i(t)=1$ if node $i$ transmits a status update to the BS in time slot $t$ and $u_i(t)=0$ otherwise. Considering error-free bandwidth-limited channel, each transmission will succeed with probability $1$ and at most $N$ nodes can transmit to the BS in each time slot. 

The age of synchronization (AoS) describes the time elapsed since the last time that the information at the BS was synchronized with the information at a node. The AoS of node $i$ at the end of time slot $t$ is defined as $s_i(t)$. If the newest status update at node $i$ has been received by the BS, $s_i(t)$ turns $0$ and keeps $0$ until a new update arrives at node $i$. Then $s_i(t)$ will keep growing by $1$ in each subsequent time slot until the new update is received by the BS. The evolution of $s_i(t)$ can be shown as:
\begin{equation}
  s_i(t+1)=\begin{cases}
    0, &u_i(t)=1 \text{ or } s_i(t)=0, \Lambda_i(t)=0;\\
    s_i(t)+1, &\text{otherwise.} \label{aos}
  \end{cases}
\end{equation}

To characterize the alterable timeliness requirement of each node, assume that there are at most $R$ possible weight values among the nodes. The weight of the node $i$ in time slot $t$ can be categorized into states $R_i(t)\in\mathcal{R}=\{1,2,\dots,R\}$. The weight state $r$ can be mapped to the actual weight value by function $\omega_i(r)$. So the weight of node $i$ in time slot $t$ is denoted by $\omega_i(R_i(t))$. Based on the actual scenarios, the sequence $\{R_i(t)\}_t$ for each node $i$ can be modeled as a Markov chain, where the next state only depends on the current state. The transition probability function can be defined as
\begin{align}
  P_{r, r'}^i = \text{Pr}\Big(R_i(t+1)=r'|R_i(t)=r\Big),& \label{pr}\\ 
  ~\forall r, r'\in \mathcal{R},& i\in\{1,2,\dots,M\}. \nonumber
\end{align}

The scheduling policy is denoted by $\pmb{\pi}\in\Pi$, where $\Pi$ is the policies set. The goal of this work is to minimize the expected weighted-sum AoS of the system. So the optimal policy $\pmb{\pi}^*$ can be derived by the following optimization problem.

\begin{pb}[Optimization problem] \label{pb1}
  \begin{subequations}
    \begin{align}
      \pmb{\pi}^*=\arg\min_{\pmb{\pi}\in\Pi}\lim_{T\to\infty}&\mathbb{E}_{\pmb{\pi}}\Biggl[\frac{1}{T}\sum_{t=1}^T\sum_{i=1}^M\omega_i\Bigl(R_i(t)\Bigr)s_i(t)\Biggr], \label{p1} \\
      \text{s.t. } \sum_{i=1}^M &u_i(t) \le N, ~\forall t \in \{1, \dots, T\}. \label{c1}
    \end{align}
  \end{subequations}
\end{pb}

The expectation in \eqref{p1} is taken with respect to the randomness of arrivals and weights.

\section{Problem Resolution} \label{resolution}
In this section, we solve Problem~\ref{pb1} in three steps. First, the original problem is relaxed to a time-average constraint. Second, the relaxed problem is formulated as a Lagrange equation and solved under a fixed Lagrange multiplier. Third, the appropriate Lagrange multiplier is found.

\subsection{Relaxed Problem}
We can formulate Problem~\ref{pb1} as a constrained Markov decision process (CMDP). The state of the CMDP is composed of the AoS $\{s_i(t)\}_{i=1}^M$ and the weight state $\{R_i(t)\}_{i=1}^M$. The action of the CMDP is the transmission status $\{u_i(t)\}_{i=1}^M$. Noting that the size of the action space grows exponentially with $N$, Problem~\ref{pb1} is hard to solve through numerical iteration methods.

To solve the original problem, we relax the bandwidth constraint in every time slot \eqref{c1} to a time-average bandwidth constraint, which gives a relaxed problem.
\begin{pb}[Relaxed problem] \label{pb2}
  \begin{subequations}
    \begin{align}
      \pmb{\pi}_{re}^*=\arg\min_{\pmb{\pi}\in\Pi}\lim_{T\to\infty}&\mathbb{E}_{\pmb{\pi}}\Biggl[\frac{1}{T}\sum_{t=1}^T\sum_{i=1}^M\omega_i\Bigl(R_i(t)\Bigr)s_i(t)\Biggr], \label{p2} \\
      \text{s.t. } \lim_{T\to\infty}&\mathbb{E}_{\pmb{\pi}}\Biggl[\frac{1}{T}\sum_{t=1}^T\sum_{i=1}^M u_i(t)\Biggr] \le N. \label{c2}
    \end{align}
  \end{subequations}
\end{pb}

The optimization objective and the constraint in Problem~\ref{pb2} can be combined by a Lagrangian equation
\begin{align}
  &\mathcal{L}(\pmb{\pi}, \eta)= \label{l1} \\ 
  &\lim_{T\to\infty}\mathbb{E}_{\pmb{\pi}}\Biggl[\frac{1}{T}\sum_{t=1}^T\sum_{i=1}^M\biggl(\omega_i\Bigl(R_i(t)\Bigr)s_i(t)+\eta u_i(t)\biggr)-\eta N\Biggr]. \nonumber
\end{align}
where $\eta>0$ is the Lagrange multiplier, which represents the strictness of the bandwidth constraint. Under a fixed $\eta$, an optimal solution $\pmb{\pi}^*(\eta)$ can be derived from \eqref{l1}. Choosing an appropriate $\eta$ gives the optimal policy for Problem~\ref{pb2} as $\pmb{\pi}_{re}^* = \pmb{\pi}^*(\eta)$.

For $M$ is finite, the order of summation over $t$ and $n$ in \eqref{l1} can be interchanged, allowing the Lagrange equation to be decomposed into equations for each individual node. The policy $\pmb{\pi}^*(\eta)$ can be decoupled as
\begin{equation}
  \pmb{\pi}^*(\eta)=\bigotimes_{i=1}^M \pi_i^*(\eta).
\end{equation}
where $\pi_i^*(\eta)$ represents the optimal policy derived from the following individual optimization problem for each node. 
\begin{pb}[Individual problem for node $i$] \label{pb3}
  \begin{align}
    &\pi_i^*(\eta)=\\ 
    &\arg\min_{\pi\in\Pi}\lim_{T\to\infty}\mathbb{E}_{\pi}\Biggl[\frac{1}{T}\sum_{t=1}^T\biggl(\omega_i\Bigl(R_i(t)\Bigr)s_i(t)+\eta u_i(t)\biggr)\Biggr]. \nonumber
  \end{align}
\end{pb}

\subsection{Problem for individual node} \label{sec3.2}
For simplicity, the node index $i$ can be omitted in this section. Define the term in the summation as cost
\begin{equation}
  C\Big(s(t),R(t),u(t)\Big) = \omega(R(t))s(t) + \eta u(t). \label{cost}
\end{equation}

Problem~\ref{pb3} can be formulated as an MDP problem with cost $C(s,R,u)$. The state is $(s(t), R(t))$ and the action is $u(t)=\{0, 1\}$. Based on the evolution of the AoS \eqref{aos} and the state transition function of the weight state \eqref{pr}, the transition probability of the MDP can be obtained as follows.

Assuming that the current state is $s(t)=s, R(t)=r$. In the states with $s>0$, the BS is not synchronized with the node and there is a new package that needs to be transmitted. The AoS of the next state will grow by $1$ if the package is not transmitted. If it is transmitted, the AoS will turn to $0$ with probability $(1-\lambda)$, meaning there is no newer package arrived, and to $1$ with probability $\lambda$, meaning there is a new arrival in this time slot. The weight state follows \eqref{pr}.
\begin{subequations}
  \begin{align}
    \text{Pr}(0, r'|s, r)=&P_{r, r'}(1-\lambda), &u(t)=1, s>0 \label{pr1}\\
    \text{Pr}(1, r'|s, r)=&P_{r, r'}\lambda, &u(t)=1, s>0 \label{pr2}\\
    \text{Pr}(s+1, r'|s, r)=&P_{r, r'}, &u(t)=0, s>0 \label{pr3}
  \end{align}
\end{subequations}

In the states with $s=0$, the AoS of the next state only depends on the arrival.
\begin{subequations}
  \begin{align}
    \text{Pr}(0, r'|0, r)=&P_{r, 'r'}(1-\lambda), \label{pr4}\\
    \text{Pr}(1, r'|0, r)=&P_{r, r'}\lambda. \label{pr5}
  \end{align}	
\end{subequations}

The state space of the MDP is infinite, making traversal infeasible. So we need to find an upper bound of the AoS to make the space finite. 
\begin{theorem} \label{th1}
  There exists an optimal stationary policy $\pi_i^*$ and a set of threshold $\{\tau_r\}_{r\in\mathcal{R}}$ such that the node transmits to the BS with probability 1 in the states with $r$ satisfying $s\ge\tau_r$ and keeps idle with probability 1 while $s<\tau_r$.
\end{theorem}

\begin{proof}
  The proof of Theorem~\ref{th1} is similar to \cite[Theorem 1]{Tang2021}. A brief description is shown as follows.

  If the optimal action is $u = 1$ in state $(s_1, r)$, it can be proven that for the state $(s, r), ~\forall s>s_1$, the optimal action is also $u=1$. Similarly, if the optimal action is $u = 0$ in state $(s_2, r)$, the optimal action in state $(s, r), ~\forall s<s_2$ is also $u=0$. Therefore, the optimal policy $\pi_i^*$ has a threshold structure. The threshold structure leads to the stability.

\end{proof}

Then we can choose a large enough $S_{max}$ as the upper bound of the AoS so that
\begin{equation}
  S_{max} \ge \tau_r, ~\forall r \in \mathcal{R}.
\end{equation}

After the state space of the MDP problem is restricted into a finite space $\{0, 1, \dots, S_{max}\}\times\mathcal{R}$, Problem~\ref{pb3} can be solved by a linear programming (LP) algorithm employed in \cite{Eitan1999}. While the state is $(s(t), R(t)) = (s, r)$, define $\mu_{s,r}$ as the steady-state distribution probability, and the occupation measure $\nu_{s,r} \le \mu_{s,r}$ as the steady-state probability of taking action $u(t)=1$. Then Problem~\ref{pb3} can be rewritten as the following LP problem
\begin{pb}[Linear programming problem for individual node] \label{pblp}
  \begin{subequations}
    \begin{align}
      &\{\mu_{s,r}^*,\nu_{s,r}^*\} = \arg\min_{\{\mu_{s,r},\nu_{s,r}\}} \sum_{s=0}^{s_{max}}\sum_{r=1}^R\Big(\omega(r)s\mu_{s,r} + \eta\nu_{s,r}\Big), \label{flp}\\
      &\text{s.t. } \sum_{s=0}^{s_{max}}\sum_{r=1}^R \mu_{s,r} = 1, \label{eq1}\\ 
      &\mu_{0,r} = \sum_{s=1}^{s_{max}}\sum_{r'=1}^R\nu_{s,r'}\text{P}_{r', r}(1-\lambda)\nonumber\\
      &\qquad +\sum_{r'=1}^R \mu_{0,r'}\text{P}_{r', r}(1-\lambda), ~\forall r\in\mathcal{R}, \label{eq2}\\
      &\mu_{1,r} = \sum_{s=1}^{s_{max}}\sum_{r'=1}^R\nu_{s,r'}\text{P}_{r', r}\lambda+\sum_{r=1}^R \mu_{0,r'}\text{P}_{r', r}\lambda, ~\forall r\in\mathcal{R}, \label{eq3}\\
      &\mu_{s,r} = \sum_{r'=1}^R (\mu_{s-1,r'}-\nu_{s-1,r'})\text{P}_{r', r},\nonumber\\
      &\qquad  ~\forall s\in\{2,\dots,S_{max} - 1\}, r\in\mathcal{R}, \label{eq4}\\
      &\mu_{S_{max},r} = \sum_{r'=1}^R (\mu_{S_{max}-1,r'} - \nu_{S_{max}-1,r'} \nonumber\\
      &\qquad + \mu_{S_{max},r'} - \nu_{S_{max},r'})\text{P}_{r', r}, ~\forall r\in\mathcal{R}, \label{eq5}\\
      % &\mu_{s_{max}, r} = \nu_{s_{max}, r}, \qquad~\forall r\in\mathcal{R}, \label{eq5}\\
      &\mu_{s,r} \ge \nu_{s,r} \ge 0, \qquad~\forall s\in\{0,\dots,S_{max}\}, r\in\mathcal{R}. \label{ie1}
      % &\mu_{s,r} \ge 0, \qquad~\forall s\in\{0,\dots,S_{max}\}, r\in\mathcal{R},\\
      % &\nu_{s,r} \ge 0, \qquad~\forall s\in\{0,\dots,S_{max}\}, r\in\mathcal{R}.
    \end{align}
  \end{subequations}
\end{pb}

Equality constraint \eqref{eq1} represents that the steady-state distribution probability of all the states adds up to $1$. Equality constraint \eqref{eq2}-\eqref{eq4} correspond to transition probability \eqref{pr1}-\eqref{pr3} respectively. The last equality constraint \eqref{eq5} is the boundary condition, which restricts the AoS from exceeding $S_{max}$. Problem~\ref{pblp} can be solved through specific linear programming algorithms.

The optimal stationary policy is defined as the probability of transmitting in each state
\begin{equation}
  \xi_{s,r}^*=\frac{\nu_{s,r}^*}{\mu_{s,r}^*}.
\end{equation}
where $\nu_{s,r}^*$ and $\mu_{s,r}^*$ is the solution of Problem~\ref{pblp}. The states with $\mu_{s,r}^*=0$ are not reachable so that the corresponding $\xi_{s,r}^*$ can be set to $1$.

\subsection{Finding appropriate Lagrange multiplier}

The optimal stationary policy $\pi_i^*(\eta)$ of the individual problem for node $i$ under given $\eta$ consists of transmission probability $\xi_{s,r}^i(\eta), s\in\{0,\dots,S_{max}^i(\eta)\}, r\in\mathcal{R}$, where $S_{max}^i(\eta)$ represents the AoS upper bound of node $i$ and $\xi_{s,r}^i(\eta)=\nu_{s,r}^i(\eta)/\mu_{s,r}^i(\eta)$. The optimal policy $\pmb{\pi}^*(\eta)$ of the Lagrange function \eqref{l1} under fixed $\eta$ can be obtained by $\{\pi_i^*(\eta)\}$. 

To find the appropriate $\eta$, rewrite \eqref{l1} as the function of $\eta$
\begin{equation}
  \mathcal{L}(\pmb{\pi}^*(\eta), \eta) = J^*(\eta) + \eta D^*(\eta) - \eta N,
\end{equation}
where $J^*(\eta)$ and $D^*(\eta)$ denote the time-average weighted sum AoS and the time-average number of transmissions.
\begin{align}
  J^*(\eta) &= \sum_{i=1}^M\sum_{s=0}^{S_{max}^i}\sum_{r=1}^R \omega(r)s\cdot\mu_{s,r}^i(\eta),\\
  D^*(\eta) &= \sum_{i=1}^M\sum_{s=0}^{S_{max}^i}\sum_{r=1}^R \nu_{s,r}^i(\eta).
\end{align}

Then the bandwidth constraint \eqref{c2} can be expressed as 
\begin{equation}
  D^*(\eta) \le N.
\end{equation}

For $\eta$ is the coefficient of $\{\nu_{s,r}\}$ in \eqref{flp}, it is easy to prove that $D^*(\eta)$ is monotonically non-increasing with respect to $\eta$. Since the action $\{u_i(t)\}$ is discrete, there may not exist an exact value of $\eta$ satisfying $D^*(\eta)=N$. Based on \cite{Tang2021}, the optimal policy $\pmb{\pi}_{re}^*$ of Problem~\ref{pb2} is a combination of at most two policies $\pmb{\pi}^*(\eta_1)$ and $\pmb{\pi}^*(\eta_2)$, where
\begin{align}
  \eta_1 &= \mathop{\arg\min}_{\eta>0}\{D^*(\eta)|D^*(\eta)\ge N\},\\
  \eta_2 &= \mathop{\arg\max}_{\eta>0}\{D^*(\eta)|D^*(\eta)\le N\}.
\end{align}

If $\eta_1$ is equal to $\eta_2$, meaning $D^*(\eta_1)=N$, the optimal policy $\pmb{\pi}_{re}^*$ is $\pmb{\pi}^*(\eta_1)$. Otherwise define the steady-state distribution probability $\mu_{s,r}^{i,*}$ and the occupation measure $\nu_{s,r}^{i,*}$ as
\begin{align}
  \mu_{s,r}^{i,*} =& \alpha\mu_{s,r}^i(\eta_1) + (1-\alpha)\mu_{s,r}^i(\eta_2),\\\nonumber
  &~\forall i\in\{1,\dots,M\},r\in\mathcal{R},s\in\{0,\dots,S_{max}^{i,*}\}, \\
  \nu_{s,r}^{i,*} =& \alpha\nu_{s,r}^i(\eta_1) + (1-\alpha)\nu_{s,r}^i(\eta_2),\\\nonumber
  &~\forall i\in\{1,\dots,M\},r\in\mathcal{R},s\in\{0,\dots,S_{max}^{i,*}\}.
\end{align}
where 
\begin{equation}
  \alpha = \frac{N-D^*(\eta_2)}{D^*(\eta_1)-D^*(\eta_2)},
\end{equation}
\begin{equation}
  S_{max}^{i,*} = \max\{S_{max}^i(\eta_1), S_{max}^i(\eta_2)\}, ~\forall i\in\{1,\dots,M\}.
\end{equation}

The optimal policy $\pmb{\pi}_{re}^*$ is the set of the transmission probability $\xi_{s,r}^{i,*} = \nu_{s,r}^{i,*}/\mu_{s,r}^{i,*}$.

The last issue is that $\pmb{\pi}_{re}^*$ cannot meet the constraint of the initial optimization problem \eqref{c1}. Define a sub-optimal policy $\pmb{\pi}_s$ which acts the same as $\pmb{\pi}_{re}^*$ when the bandwidth constraint is met, but randomly selects $N$ nodes to transmit in those to be transmitting under $\pmb{\pi}_{re}^*$ when $\pmb{\pi}_{re}^*$ cannot satisfy the bandwidth constraint. Then $\pmb{\pi}_s$ is our policy for the original problem. Though $\pmb{\pi}_s$ is not stationary for a random selection may be made by the central controller, the transmission requests of each node are stationary. We can consider $\pmb{\pi}_s$ as a near-stationary policy.

Define the time-average weighted sum AoS of $\pmb{\pi}_s$, $\pmb{\pi}_{re}^*$ and $\pmb{\pi}^*$ as $J_s$, $J_{re}$ and $J^*$. For Problem~\ref{pb2} is relaxed from Problem~\ref{pb1}, we have $J_s>J^*>J_{re}$. Referring to \cite[Theorem 3]{Tang2021}, if $R$ is finite and each node has the same functions $P_{r, r'}^i$ and $\omega_i(r)$, the relationship of the AoS under a fixed ratio $\theta=N/M$ satisfies
\begin{equation}
  \lim_{\substack{N\to\infty \\\theta=N/M}} \frac{J_s - J^*}{J^*} = 0.
\end{equation}

So that the expected AoS of $\pmb{\pi}_s$ converges to the minimum AoS of Problem~\ref{pb1}, as $N/M$ fixed and $N\to\infty$.

\section{Simulation Results} \label{simulation}
In this section, we provide simulation results for the AoS performance of our policy.

The number of nodes $M$ is set to $100$, and the arrival rates $\{\lambda_i\}$ uniformly distributed from $0.9$ to $0.1$. The number of weight states is $R=2$. Referring to \cite{Tang2021}, the weight of node $i$ in state $r$ is given by the weight function $\omega_i(r) = c_i\cdot o_r$, where $c_i = 1/i^\delta$ follows a Zipf distribution with coefficient $\delta = 0.7$ and $\{o_1, o_2\}$ is $\{1, 10\}$. The weight state transition function of each node is
\begin{align}
  P = \begin{bmatrix}
    q & 1-q \\
    1-q & q
  \end{bmatrix},
\end{align}
where $q\in(0, 1)$ is the self-transition probability.

Figure~\ref{q} shows the AoS under different self-transition probability $q$. The bandwidth limitation, meaning the number of nodes allowed to transmit in a time slot, is set to $N=30$ and the total simulation time is $T=10^5$. For comparison, we simulate the Max-Weight policy which is proposed from \cite{He2022}. The lower bound is the AoS of $\pmb{\pi}_{re}^*$, the policy under relaxed bandwidth constraint \eqref{c2}. The simulation results indicate that our policy shows benefits over the Max-Weight in the points away from $q=0.5$. The AoS performance of two policies is close when $q$ approaches $0.5$, because the next weight state is nearly independent of the current state and the information obtained from state transition is small. The farther $q$ deviates from $0.5$, our policy performs better for greater gains from state transition information. 

\begin{figure}[h]
  \centering
  \includegraphics[width=0.76\columnwidth]{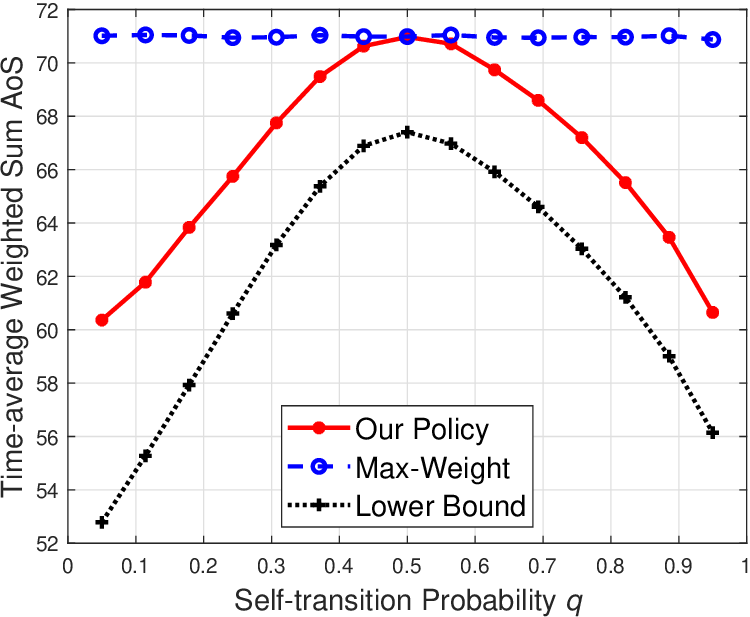}
  \caption{The time-average weighted sum AoS varies with self-transition probability.} 
  \label{q}
\end{figure}

Then we choose $q=0.1$ to compare the AoS performance under different bandwidth limitation $N$ in scenarios with non-independent state transitions. As shown in Figure~\ref{N}, our policy outperforms the Max-Weight under different number of nodes allowed to transmit in a time slot. The advantage is more pronounced under looser bandwidth limitation, because more sufficient resources provide greater flexibility in the application of our policy. The gap between our policy and the lower bound reduces under bigger $N$, meaning that the relaxed constraint \eqref{c2} approaches the original \eqref{c1}.

\begin{figure}[h]
  \centering
  \includegraphics[width=0.76\columnwidth]{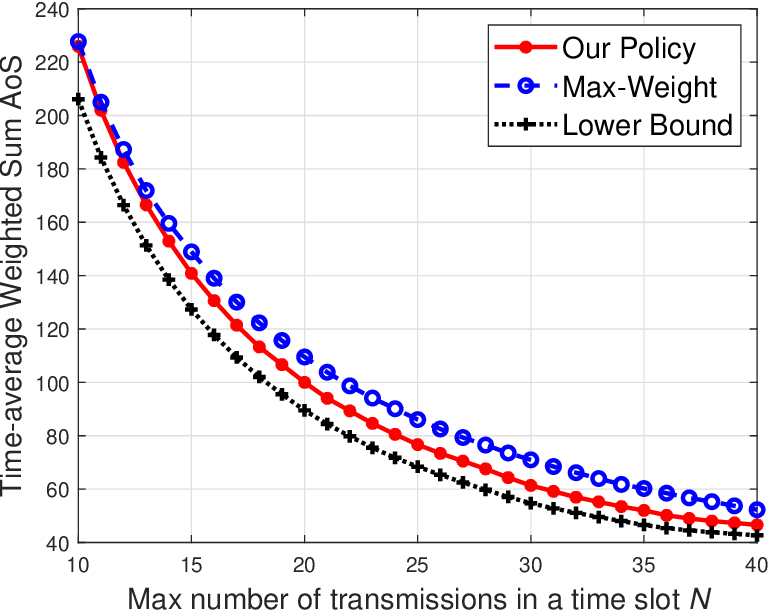}
  \caption{The time-average weighted sum AoS varies with bandwidth limitation.} 
  \label{N}
\end{figure}

\section{Conclusion} \label{conclusion}
In this study, we investigated the optimization problem of information freshness in a wireless network where randomly arrived status updates are transmitted from multiple nodes to a base station. The information freshness of the updates is characterized by the Age of Synchronization (AoS). The nodes share an error-free, bandwidth-limited channel and have different time-varying weights reflecting the importance of the AoS. The variations of the weights are modeled as the Markov chain. Although an error-free channel does not meet general conditions, the study can be easily extended to error-prone channels, since the only difference lies in the transition probability. The goal of this study is to design a scheduling policy minimizing the weighted sum AoS of the network. The optimization problem is relaxed and described as a CMDP, which is decoupled into problems about the individual node. Then the problems are solved by the Lagrange function and the linear programming algorithm, leading to a stationary policy of the relaxed problem. We proposed the policy for the original problem by operating the stationary policy under the initial bandwidth constraint. This policy is near-stationary and asymptotically optimal under specific assumptions. Numerical simulations show that the AoS performance of our policy has a clear advantage over the Max-Weight policy under non-independent weight state transitions.

\bibliography{reference}
\bibliographystyle{IEEEtran}

\end{document}